\documentclass[a4paper,12pt]{article}
\usepackage{latexsym,amssymb,amsfonts,amsmath,amsthm}
\usepackage[dvips]{graphicx}
\newtheorem{thm}{Theorem}[section]
\newtheorem{defini}[thm]{Definition}
\newtheorem{lem}[thm]{Lemma}
\newtheorem{cor}[thm]{Corollary}
\newtheorem{pro}[thm]{Proposition}
\newtheorem{exa}[thm]{Example} 

\title{{\Large {\bf 
A remark on zeta functions of finite graphs via quantum walks
}
}}
\author{ 
{\small 
Yusuke HIGUCHI,$^{1}$ 
\quad 
Norio KONNO,$^{2}$ 
\quad
Iwao SATO,$^{3}$ 
\quad
Etsuo SEGAWA$^{4}$ 
}\\ 
{\scriptsize $^{1}$ 
Mathematics Laboratories, College of Arts and Sciences, Showa University
}\\
{\scriptsize 
4562 Kamiyoshida, Fujiyoshida, Yamanashi 403-0005, Japan
} \\
{\scriptsize $^3$ 
Department of Applied Mathematics, Faculty of Engineering, Yokohama National University 
}\\
{\scriptsize 
Hodogaya, Yokohama 240-8501, Japan
} \\
{\scriptsize $^3$ 
Oyama National College of Technology
}\\
{\scriptsize 
Oyama, Tochigi 323-0806, Japan
} \\
{\scriptsize $^4$ 
Graduate School of Information Sciences, Tohoku University
}\\
{\scriptsize
Sendai 980-8579, Japan@
}\\
} 
\date{\empty }
\begin{document}
\maketitle

\par\noindent
\begin{small}
\par\noindent
{\bf Abstract}. 
From the viewpoint of quantum walks, 
the Ihara zeta function of a finite graph can be said to be 
closely related to its evolution matrix. 
In this note we introduce another kind of zeta function of a graph, 
which is closely related to, as to say, the square of the evolution 
matrix of a quantum walk. 
Then we give to such a function two types of determinant expressions 
and derive from it some geometric properties of a finite graph. 
As an application, we illustrate the distribution of poles of this function 
comparing with those of the usual Ihara zeta function.


\end{small}

\setcounter{equation}{0}

\section{Introduction}
As is the classical random walk on a graph  
has important roles in various fields, 
the quantum walk, say QW, is expected to play such a role in the quantum field. 
In fact, we can find many studies on QW cover a wide research area 
from the basic theoretical mathematics to the application oriented fields. 
It has been shown, for example, that 
analyzing some spatial structure 
\cite{AmbainisEtAl2001,segawa4,segawa3} 
as an extension of quantum speed-up algorithm \cite{Grover1,Grover2}, 
and application 
to a universal computation in quantum mechanical computers \cite{segawa5}, 
expressing the energy transfer on the chromatographic network 
in the photosynthetic system \cite{segawa6} and so on are 
strongly influenced by its virtue. 
Besides, approximations of QWs describing physical processes are derived from 
 Dirac and Schr\"odinger equations \cite{segawa8,segawa7}. 
A QW model has been also shown to be useful for describing 
the fundamental dynamics of the quantum multi-level system which is 
irradiated by lasers \cite{segawa17}. 
The laser control technology of quantum system is 
expected to be applied for the industry as a highly-selective method 
for material separation, especially, isotope-selective excitation 
of diatomic molecules such as Cs133 and Cs135. 
Recently by the above theoretical evidences for the usefulness and 
activeness of the studies of QWs, experimental implementations of QWs are 
quite aggressively investigated. 
See \cite{segawa18,segawa19,segawa20}, for example.   

Now we shall focus on mathematical research on QW. 
The starter creating studies of QW in earnest are considered as 
the QW on one dimensional lattice introduced by \cite{AmbainisEtAl2001}:  
one of the most striking properties  is the spreading property of the walker. 
Its standard deviation of the position grows linearly in time, 
quadratically faster than the classical random walk. 
The behaviour is clarified by a limit theorem characterized by 
a new density function named ``$f_{K}$ function'' 
\cite{Konno2002,Konno2005}. 
The review and book on QWs are 
J.~Kempe \cite{Kempe2003} and N.~Konno \cite{Konno2008b}. 
See also \cite{Ambainis2003, manoucheheri2014,VA}. 
For a general graph, it is usual to consider some special but typical type 
of QWs:  
the {\em Grover walk} originated in \cite{Grover1,Grover2} or 
the {\em Szegedy walk} in \cite{segawa3}. Roughly speaking, 
the former is induced by the simple random walk and the latter by 
more general random walk on a graph. 
In this context, the relationship between spectra of QW and 
that of the classical random walk is investigated in 
\cite{EmmsETAL2006,HKSS,KS2012,Segawa}. 
From now on 
we call the evolution matrices of the Grover walk and the Szegedy walk 
just the {\em Grover matrix} ${\bf U}$   
and the {\em Szegedy matrix} ${\bf U}_{sz}$, respectively. 

Recently there are some trials to apply QW to graph isomorphism problems 
\cite{EmmsETAL2006,EmmsETAL2009,GambelETAL,ShiauETAL}. 
For graph isomorphism problems, while spectra of the Grover matrix 
is considered to have 
almost same power as that of conventional operator,  
it is suggested that 
the method of $({\bf U}^{3})^{+}$, which is the {\em positive support} of 
the cube of the Grover matrix ${\bf U}^{3}$, outperforms 
the graph spectra methods, in particular,  
in distinguishing strongly regular graphs in \cite{EmmsETAL2006}. 
What we emphasize is that not only the Grover matrix ${\bf U}$ itself 
but the {\em positive support} $({\bf U}^{n})^{+}$ of its $n$-th power 
is an important operator of a graph. 
See also \cite{GG2010, HKSS}. Meanwhile, 
in \cite{KS2012,RenETAL} the relationship between the Ihara zeta function 
and the positive support $({\bf U})^{+}$ of the Grover matrix of a graph 
is discussed: a matrix $({\bf U})^{+}$ derived from QW is essentially the same 
as the edge-matrix in \cite{Bass1992,Hashimoto1989} 
and the Perron-Frobenius operator 
in \cite{KS2000}, both of which are important operators in characterizing 
that function. The Ihara zeta functions of graphs started for regular 
graphs by Y.~Ihara \cite{Ihara1966} and is generalized to a general graph. 
Already various success related to graph spectra is 
obtained in \cite{Bass1992,Hashimoto1989,Ihara1966, KS2000, Sunada1986}. 

This note is a sequel work to our previous work \cite{HKSS}, 
therein we established a general relation 
between QW and the classical random walk; as its application, 
we recover the results in \cite{EmmsETAL2006,GG2010,KS2012,Segawa} 
of spectral relation between three matrices 
${\bf U}$, $({\bf U})^{+}$, $({\bf U}^{2})^{+}$ 
from QW and 
the adjacency matrix ${\bf A}_{G}$. 
Our main purpose in this note is 
to characterize another kind of zeta function with respect to 
$({\bf U}^{2})^{+}$, which is 
the positive support of the squared Grover matrix. 

To state our result precisely, let us give our setting.
A graph $G$ is a pair of two sets $(V(G), E(G))$, where $V(G)$ stands for 
the set of its vertices and $E(G)$ the set of its unoriented edges. 
Assigning two orientations to each unoriented edge in $E(G)$, 
we introduce the set of all oriented edges and denote it by $D(G)$. 
For an oriented edge $e\in D(G)$, the origin of $e$, the terminus of $e$ 
and the inverse edge of $e$ are denoted by $o(e)$, $t(e)$ and $e^{-1}$,
respectively. Furthermore the {\em degree} of $x\in V(G)$, $\deg_{G} x$, 
is defined as the number of oriented edges $e$ such that $o(e)=x$; 
we denote $\min_{x\in V(G)}\deg_{G} x$ and $\max_{x\in V(G)}\deg_{G} x$ 
by $\delta (G)$ and $\Delta (G)$, respectively. 
A graph $G$ here is basically assumed to be a connected finite graph 
with $n$ vertices, $m$ unoriented edges and $\delta (G)\geq 3$; 
it may have multiple edges or self-loops. 
For a natural number $k$, if $\deg_{G}v=k$ for each vertex $v\in V(G)$, 
then a graph $G$ is called {\em $k$-regular}. 

Let us introduce the {\em Grover matrix} ${\bf U}_{G}={\bf U}$, 
which is a special QW related to the simple random walk on $G$, 
and the {\em positive support} ${\bf F}^{+}$ for a real matrix ${\bf F}$.  

\begin{defini}
The {\em Grover matrix} ${\bf U}=( U_{e,f} )_{e,f \in D(G)} $ 
of $G$ is a $2m\times 2m$ matrix defined by 
\[
U_{e,f} =\left\{
\begin{array}{ll}
2/\deg_{G} o(e),  & \mbox{if $t(f)=o(e)$ and $f \neq e^{-1} $, } \\
2/\deg_{G} o(e) -1, & \mbox{if $f= e^{-1} $, } \\
0, & \mbox{otherwise},
\end{array}
\right. 
\]
and the {\em positive support} ${\bf F}^+ =( F^+_{i,j} )$ of 
a real square matrix ${\bf F} =( F_{i,j} )$ is defined by 
\[
F^+_{i,j} =\left\{
\begin{array}{ll}
1, & \mbox{if $F_{i,j} >0$, } \\
0, & \mbox{otherwise}. 
\end{array}
\right.
\] 
\end{defini}

Properties of the Grover matrix can be seen in \cite{Grover1,Grover2}; 
see also \cite{EmmsETAL2006,GG2010,HKSS,KS2012,Segawa}. 
The {Szegedy matrix} related to a 
general random walk on $G$ is omitted since we do not use here; 
its definition and properties can be seen in 
\cite{HKSS, Segawa, segawa3}, for instance. 
The spectra of 
the positive support ${\bf U}^{+} $ of the Grover matrix    
and  $({\bf U}^{2})^{+} $ of its square 
on a regular graph $G$ are expressed in \cite{EmmsETAL2006}, 
also in \cite{GG2010, HKSS}, 
by means of those of the {\em adjacency matrix} ${\bf A}_{G}$ of $G$, 
which is an important matrix also in this note and defined as follows:  
the {\em adjacency matrix} ${\bf A}_{G} = (a_{x,y})_{x,y\in V(G)}$ is an  
$n\times n$-matrix such that 
$a_{x,y}$ coincides with the number of oriented edges such that  
$o(e)=x$ and $t(e)=y$. 

Now let us consider the following function ${\bf Z}_{G} (u)$ of a graph $G$  
for $u \in \mathbb{C}$ with $|u|$ sufficiently small: 
\begin{equation}\label{generalzeta}
{\bf Z}_{G} (u)= \prod_{[C]} (1- u^{ \mid C \mid } )^{-1}. 
\end{equation}

In (\ref{generalzeta}), if 
$[C]$ runs over all {\em equivalence classes of prime and reduced cycles} 
of $G$, then ${\bf Z}_{G} (u)$ becomes the well-known {Ihara zeta function}. 
Details will be seen Section~2, therein we give a brief summary on 
the Ihara zeta function. Roughly speaking, we will find 
two matrices $({\bf U})^{+}$ and ${\bf A}_{G}$ control this function. 
On the other hand, if 
$[C]$ runs over all {\em equivalence class of prime 2-step-cycles} 
of $G$, then ${\bf Z}_{G} (u)$ becomes a {\em modified zeta function}, 
say $\tilde{\bf Z}_{G} (u)$, which is the main object in this note. 
Precise definitions around this can be seen in Section~3. 
Roughly speaking, we will find two matrices $({\bf U}^{2})^{+}$ 
and ${\bf A}_{G}$ control this function. 
Our main theorem in this note is as follows: 

\begin{thm}[Main Theorem]
Let $G$ be a simple connected graph with $n$ vertices, $m$ unoriented edges 
and $\delta (G)\geq 3$. Then 
\begin{eqnarray}
\tilde{{\bf Z}}_{G}(u) &=& 1/\det ( {\bf I}_{2m} -u ({\bf U}^2)^+ ),\nonumber\\
&=& (1-2u )^{2(n-m)}\cdot (p_{G}(u))^{-1},\nonumber
\end{eqnarray} 
and $p_{G}(1/2)=0$. If $G$ is not bipartite, 
then the derivative at $u=1/2$ of $p_{G}(u)$ is as follows: 
\[
p'_{G} (1/2)= \frac{m-n}{2^{2n-2}} \cdot \kappa (G)\cdot\iota(G), 
\] 
where $\kappa (G)$ is the number of spanning trees in $G$ 
and $\iota (G)$ is the following graph invariant:
\[
\iota (G) =\sum_{H\in OUCF(G)}4^{\omega (H)}.
\]
Here $OUCF(G)$ stands for the set of all {\em odd-unicyclic factors} in $G$. 
On the other hand, if $G$ is bipartite, then $p'(1/2)=0$ 
and the second derivative at $u=1/2$ is as follows:
\[
p''_{G} (1/2)= \frac{(m-n )^{2} }{2^{2n-5}} ( \kappa (G))^{2}.
\]
Furthermore $u=\rho$ is also a pole, whose order $2$ or $1$ if 
$G$ is bipartite or not, respectively. 
Here $\rho$ is the radius of convergence of (\ref{generalzeta}). 
\end{thm}

Definitions not given here and details can be seen in Section~3,  
especially in  Proposition~\ref{PSDE}, Theorems~\ref{DE2} and~\ref{pole}. 
Also the radius of convergence 
is discussed in Theorem~\ref{RC}. 

The rest of the paper is organized as follows. 
In Section~2,  we present a brief survey on the Ihara zeta function 
${{\bf Z}}_{G}(u)$  of a graph, which is related to $({\bf U})^{+}$. 
In Section~3, we introduce and discuss 
a modified zeta function $\tilde{{\bf Z}}_{G}(u)$
related to $({\bf U}^{2})^{+}$ on a graph $G$ and 
present two types of determinant expressions, properties of poles and 
geometric information derived from $\tilde{{\bf Z}}_{G}(u)$. 
In Section~4, we illustrate  
the distribution of poles of $\tilde{{\bf Z}}_{G}(u)$ for a $k$-regular graph 
comparing with those of the Ihara zeta function. 
\section{The Ihara zeta function via QW} 
In this section, we shall summarize the results on the 
Ihara zeta function of a graph. 

Let $G$ be a connected graph. 
A {\em closed path} or {\em cycle} of length $\ell$ in $G$ is a sequence  
$C=(e_{0}, \dots ,e_{\ell-1} )$ of $\ell$ oriented edges 
such that $e_{i}\in D(G)$  and $t( e_{i} )=o( e_{i+1} )$ 
for each $i\in\mathbb{Z}/\ell\mathbb{Z}$.  
Such a cycle is often called an $o(e_{0})$-cycle. 
We say that a path $P=(e_{0}, \cdots ,e_{\ell-1} )$ has a {\em backtracking} 
if $ e^{-1}_{i+1} =e_{i} $ for some $i\in\mathbb{Z}/\ell\mathbb{Z}$.  
The {\em inverse cycle} of a cycle 
$C=( e_{0}, \cdots ,e_{\ell-1} )$ is the cycle 
$C^{-1} =( e^{-1}_{\ell-1} , \cdots ,e^{-1}_{0} )$. 

We introduce an equivalence relation between cycles. 
Two cycles $C_{1}$  and $C_{2}$ are said to be {\em equivalent} if  
$C_{1}$ can be obtained from $C_{2}$ by a cyclic permutation of oriented edges. 
Remark that the inverse cycle of $C$ is in general not equivalent to $C$. 
Thus we write $[C]$ for the equivalence class which contains a cycle $C$. 
Let $B^{r}$ be the cycle obtained by going $r$ times around a cycle $B$:  
such a cycle is called a {\em power} of $B$. 
Furthermore, a cycle $C$ is {\em prime} if it is not a power of 
a strictly smaller cycle. 
Besides, A cycle $C$ is called {\em reduced} if $C$ has no backtracking. 
Note that each equivalence class of prime and reduced cycles of a graph $G$ 
corresponds to a unique conjugacy class of 
the fundamental group $ \pi_{1} (G,v)$ of $G$ at a vertex $v\in V(G)$. 

The {\em Ihara zeta function} of a graph $G$ is 
a function of $u \in \mathbb{C}$ with $|u|$ sufficiently small, 
defined by 
\[
{\bf Z}_{G} (u)= \prod_{[C]} (1- u^{ \mid C \mid } )^{-1} ,
\]
where $[C]$ runs over all equivalence classes of prime and reduced cycles 
of $G$ and $|C|$ is the length of a cycle $C$. 
This function ${\bf Z}(G,u)$ can be expressed as
\[
{\bf Z}_{G}(u)= \exp \biggl(\sum_{k \geq 1} \frac{N_k }{k} u^{k}\biggr), 
\]
where $N_k$ is the number of all reduced cycles of length $k$ in $G$. 
A simple proof and an estimate for the radius of convergence for 
the power series in the above can be seen, for instance, in \cite{KS2000}.  
The following determinant expression is originally given in \cite{Hashimoto1989}; other proofs are seen in \cite{Bass1992,KS2000}. 
We should remark ${}^{T}({\bf U})^{+}$, 
the transposed matrix of $({\bf U})^{+}$, 
 is essentially the same as 
the edge-matrix in \cite{Bass1992,Hashimoto1989} 
and the Perron-Frobenius operator in \cite{KS2000}. 

\begin{thm}\label{DEUZ}
(\cite{Hashimoto1989}; cf.\cite{Bass1992,Ihara1966,KS2012,KS2000,North})
Let $G$ be a connected graph with $n$ vertices and $m$ unoriented edges. 
Then the reciprocal of the Ihara zeta function of $G$ is given by 
\begin{eqnarray}
{\bf Z}_{G}(u)^{-1} & = & \det ( {\bf I} -u ({\bf U})^{+})\nonumber \\
 & = & (1- u^2 )^{m-n} f_{G}(u).\nonumber
\end{eqnarray}  
Here we put 
\[
f_{G}(u)=\det ( {\bf I}_{n} -u {\bf A}_{G}+ u^2 ({\bf D}_{G} -{\bf I}_{n} )), 
\]
where ${\bf A}_{G}$ is the adjacency matrix of $G$ 
and ${\bf D}_{G} =(d_{x,y})_{x,y\in V(G)}$ is the {\em degree matrix} of $G$ 
which is a diagonal matrix 
with $d_{x,x} = \deg_{G} x$ for  $x\in V(G)$. 
In addition, $u=1$ is a pole of ${\bf Z}_{G}(u)$ of order $m-n+1$ and 
the derivative of $f_{G}(u)$ at $u=1$ is expressed by a graph invariant 
$\kappa (G)$:  
\[
f'_{G} (1)=2(m-n) \kappa (G) , 
\]
where $\kappa (G)$ is the number of spanning trees in $G$. 
\end{thm}
The invariant $\kappa (G)$ is called the {\em complexity} of $G$ and  
the complexities for various graphs are found in \cite{Biggs,CDS}. 
Seeing the determinant expression in the above, 
 we may say the Ihara zeta function ${\bf Z}_{G}(u)$ of a graph 
is derived by the positive support $({\bf U})^{+}$ 
of the Grover matrix ${\bf U}$. 
\section{A modified zeta function via QW} 

In this section, we will discuss a modified zeta function of 
a graph with respect to the positive support of the {\em square\/} of 
the Grover matrix. 

First of all, let us introduce a new notion of cycle in a graph 
with respect to $({\bf U}^{2})^{+}$.  
For a connected graph $G$, 
a {\em 2-step-cycle} $\tilde{C}$ of length $\ell$
 in $G$ is a sequence 
$\tilde{C}=(e_{0}, \cdots ,e_{\ell -1} )$ of $\ell$ oriented edges 
 such that every ordered pair ($e_{i}$, $e_{i+1}$) is 
a {\em 2-step-arc} or a {\em 2-step-identity} 
for each $i\in \mathbb{Z}/\ell\mathbb{Z}$. 
Here a {\em 2-step-arc} $(e,f)$ is defined as follows:
there exists an oriented edge $g(\not = e^{-1} ,f^{-1})$ such that 
$o(g)=t(e)$ and $t(g)=o(f)$; a {\em 2-step-identity} $(e,f)$ is defined as 
$e=f$. Remark that a 2-step-cycle $\tilde{C}$ of length $1$ exists 
if $\tilde{C}=(e)$. 
It can be easily checked that $({}^{T}({\bf U}^{2})^{+})_{e,f}=1$ if and only if
$(e,f)$ is a {\em 2-step-arc} or a {\em 2-step-identity}.  

Similarly to the case of usual cycles in Section~2, 
we give an equivalence relation between 2-step cycles. 
Two cycles $\tilde{C}_{1}$  and $\tilde{C}_{2}$ 
are said to be {\em equivalent} if  
$\tilde{C}_{1}$ can be obtained from $\tilde{C}_{2}$ 
by a cyclic permutation of oriented edges. 
Thus we write $[\tilde{C}]$ for the equivalence class 
which contains a 2-step-cycle $\tilde{C}$. 
Let $\tilde{B}^{r}$ be the 2-step-cycle obtained 
by going $r$ times around some 2-step-cycle $B$;  
a 2-step-cycle $\tilde{C}$ is {\em prime} if it is not a multiple of 
a strictly smaller 2-step-cycle. 

Let $G$ be a connected graph 
with $n$ vertices, $m$ unoriented edges and $\delta (G)\geq 3$. 
Now let us define another kind of zeta function of a graph 
related to $({\bf U}^2)^{+}$. 

\begin{defini}
The {\em modified zeta function} of a graph $G$ is 
a function of $u \in \mathbb{C}$ with $|u|$ sufficiently small, 
defined by 

\[
\tilde{{\bf Z}}_{G}(u)=\prod_{[\tilde{C}]} (1- u^{|\tilde{C}|})^{-1} ,
\]
where $[\tilde{C}]$ is the equivalence class of prime 2-step-cycles and 
$|\tilde{C}|$ is the length of a 2-step-cycle $\tilde{C}$. 
\end{defini}
From the definitions of a {\em 2-step-cycle} and an equivalence class, 
applying the usual method, which can be seen in \cite{KS2000,Sm} for instance, 
we can give the exponential expression 
and a determinant expression 
for the modified zeta function $\tilde{{\bf Z}}_{G}(u)$:

\begin{pro}\label{PSDE}
Let $G$ be a connected graph with $n$ vertices and $m$ unoriented edges. 
Suppose that $\delta (G)\geq 3$. Then 
\begin{eqnarray}
\tilde{{\bf Z}}_{G}(u) 
&=& \exp \biggl(\sum_{r\geq 1}\frac{\tilde{N}_{r}}{r}u^{r} \biggr) 
\label{PS}\\
&=& 1/\det ( {\bf I}_{2m} -u ({\bf U}^2)^+ ),\nonumber
\end{eqnarray} 
where $\tilde{N}_r$ is the number of all 2-step-cycles of length $r$. 
\end{pro}
Now let us give estimation of the radius of convergence $\rho$ 
of the power series in the above. Naturally, $\rho$ is also 
the singular point of $\tilde{{\bf Z}}_{G}(u) $ nearest to the origin.
Recall $\delta (G)$ and $\Delta (G)$ stand for 
$\min_{x\in V(G)}\deg_{G} x$ and $\max_{x\in V(G)}\deg_{G} x$, respectively. 
\begin{thm}\label{RC}
Let $G$ be a connected graph with $\delta (G)\geq 3$.
The radius of convergence $\rho $ of the power series (\ref{PS}) in 
Proposition~\ref{PSDE} is $\rho = 1/ \alpha $, where $\alpha$ is 
the maximal eigenvalue of $({\bf U}^{2})^{+}$; it holds that 
\[
1/((\delta (G)-1)^{2} +1) \leq \rho \leq 1/( (\Delta (G)-1)^{2} +1). 
\]
In particular, 
$\tilde{{\bf Z}}_{G}(u)$ is a rational function of $u$ with a pole $\rho$
 whose order is 2 or 1 if $G$ is bipartite or not, respectively.
\end{thm}

\begin{proof}
As is seen above, $({\bf U}^{2})^{+}$ is nonnegative, that is, 
all elements are nonnegative, and $(({\bf U}^{2})^{+})_{e,f}=1$ 
if and only if $(f,e)$ is a 2-step-arc or a 2-step-identity. 
To apply the Perron-Frobenius theorem, let us discuss the irreducibility 
of $({\bf U}^{2})^{+}$. 
A matrix $M$ is called {\em irreducible} if, for each two indices $i$ and $j$,  
there exists a positive integer $k$ such that $(M^{k})_{i,j}\not= 0$. 
For the matrix $({\bf U}^{2})^{+}$, it is sufficient to see whether, 
for any two oriented edges $e,f\in D(G)$, $e$ is reachable or not from $f$ by 
an {\em admissible} sequence of 2-step-arcs and 2-step-identities, that is, 
a sequence of oriented edges $(e_{0},e_{1},e_{2},\dots ,e_{s-1},e_{s})$ 
such that $e_{0}=f$, $e_{s}=e$ and 
$(e_{k},e_{k+1})$ is a 2-step-arc or a 2-step-identity for $i=0,\dots ,s-1$. 
It is easily checked that 
such an admissible sequence from $f$ to $e$ exists if and only if 
there exists a reduced path from $f$ to $e$ of odd length 
in $G$EŒsay an {\em admissible odd path}. 
Recall that a reduced path from $e_{1}$ to $e_{\ell}$ of length $\ell$ in $G$ 
is a sequence $P=(e_{1}, \dots ,e_{\ell} )$ of $\ell$ oriented edges 
such that $t( e_{i} )=o( e_{i+1} )$ and 
$ e^{-1}_{i+1} \not= e_{i} $ for each $i=1,\dots , \ell-1$.  
Since a graph $G$ is finite and connected with $\delta (G)\geq 3$, 
$G$ has at least two {\em unoriented cycles}. 
The terminology {\em unoriented cycle} used here is the same 
as ``cycle'' in usual graph theory, that is, 
 if $C$ is an unoriented cycle of length $\ell$, then 
$V(C)=\{v_{1},\dots , v_{\ell}\}$ whose elements are mutually distinct, 
$v_{i}v_{i+1}\in E(G)$ for $i=1,\dots ,\ell-1$ and $v_{\ell}v_{1}\in E(G)$. 
For two vertices $x,y\in V(G)$, 
we denote by $dist(x,y)$ the length of the shortest path from $x$ to $y$. 
\par
For two oriented edges $e,f\in D(G)$ such that $dist(t(f),o(e))$ is odd, 
we can find an admissible odd path from $f$ to $e$. 
In particular, if $G$ is not bipartite, then $G$ has at least one 
{\em unoriented cycle} of odd length and of even length, respectively. 
Hence $G$ turns out to have  
an admissible odd path between $e$ and $f$ for any $e,f\in D(G)$; 
this implies $({\bf U}^{2})^{+}$ is irreducible. 
Next we assume $G$ is bipartite; the length of any cycle in $G$ is even. 
So we set the bipartition $V_{0}$ and $V_{1}$: $V(G)=V_{0}\sqcup V_{1}$. 
It is obvious that an admissible odd path between $e$ and $f$ exists  
if and only if both $o(e)$ and $o(f)$ in the same set of bipartition, that is, 
$o(e),o(f)\in V_{0}$ or $o(e),o(f)\in V_{1}$. 
Thus $({\bf U}^{2})^{+}$ is not irreducible and we may express, 
after rearranging rows and columns if necessary,
\[
({\bf U}^{2})^{+}=
\left(
  \begin{array}{c|c}
    {\bf M}_{0} & {\bf 0}     \\ \hline
      {\bf 0}    & {\bf M}_{1} 
  \end{array}
  \right),
\]
where ${\bf M}_{0}$ and ${\bf M}_{1}$ are $m\times m$ irreducible submatrices 
of $({\bf U}^{2})^{+}$ induced by $D_{0}=\{e\in D(G);o(e)\in V_{0}\}$ and 
$D_{1}=\{e\in D(G);o(e)\in V_{1}\}$, respectively. 
Remark that $e\in D_{0}$ if and only if $e^{-1}\in D_{1}$ and that 
an admissible sequence from $f$ to $e$ exists if and only if 
that from $e^{-1}$ to $f^{-1}$ does. Thus the characteristic polynomials 
of ${\bf M}_{0}$ and ${\bf M}_{1}$ coincide. 
\par
Now let us apply 
the Perron-Frobenius Theorem on irreducible nonnegative matrices(see \cite{Ga,GR}). 
If $G$ is not bipartite, then $({\bf U}^{2})^{+}$ has at least one positive 
eigenvalue and the maximal positive eigenvalue $\alpha$ is simple.
If $G$ is bipartite, then each of 
${\bf M}_{0}$ and ${\bf M}_{1}$ has at least one positive eigenvalue and 
simple maximal eigenvalue. 
This implies the maximal eigenvalues 
of ${\bf M}_{0}$ and ${\bf M}_{1}$ coincide, say $\alpha$. 
Hence $({\bf U}^{2})^{+}$ has the maximal eigenvalue which is positive and 
whose multiplicity is $2$ when $G$ is bipartite.  
In either case, the maximal eigenvalue $\alpha$ is estimated as follows:
\[
\min_{e\in D(G)}\sum_{f\in D(G)}(({\bf U}^{2})^{+})_{e,f}
\leq \alpha \leq \max_{e\in D(G)}\sum_{f\in D(G)}(({\bf U}^{2})^{+})_{e,f}.
\] 
It should be noted that 
the value $\sum_{f\in D(G)}(({\bf U}^{2})^{+})_{e,f}$ is equal to the number of 
$f$ such that $(f,e)$ is a 2-step-arc or a 2-step-identity for $e$. 
Then we have 
\[
(\delta (G)-1)^{2} +1 \leq \alpha \leq (\Delta (G)-1)^{2} +1. 
\] 
It is obvious to see the power series (\ref{PS}) in 
Proposition~\ref{PSDE} converges absolutely in $|u|<1/\alpha=\rho$ 
since $\tilde{N}_{r} = {\rm trace} [( ( {\bf U}^{2} )^{+} )^{r}]$.  
\end{proof}

Corresponding to Theorem~\ref{DEUZ}, another determinant expression for  
this zeta function $\tilde{{\bf Z}}_{G}(u)$ can be obtained. 
Here and hereafter we assume $G$ is {\em simple}, 
that is, $G$  has no multiple edges and no self-loops. 
\begin{thm}\label{DE2}
Let $G$ be a simple connected graph 
with $n$ vertices and $m$ unoriented edges.  
Suppose that $\delta (G)\geq 3$.  
Then the reciprocal of the modified zeta function of $G$ is given by 
\[
\tilde{{\bf Z}}_{G}(u)^{-1} 
=(1-2u )^{2(m-n)}\cdot h_{G}(u)\cdot l_{G}(u),
\]
where 
\begin{eqnarray}
h_{G}(u) &=& 
\det ( {\bf I}_{n} -\sqrt{u(1-u)} {\bf A}_{G}+u ({\bf D}_{G}-2{\bf I}_{n})), 
\nonumber\\
l_{G}(u) &=& 
\det ( {\bf I}_{n} +\sqrt{u(1-u)} {\bf A}_{G}+u ({\bf D}_{G}-2{\bf I}_{n})) 
\nonumber 
\end{eqnarray}
and ${\bf A}_{G}$ and ${\bf D}_{G}$ are, as are seen in Theorem~\ref{DEUZ}, 
the adjacency and degree matrices, respectively. 
Here two values $\sqrt{u(1-u)}$ 
in $h_{G}(u)$ and $l_{G}(u)$ are assumed to be on the same branch. 
\end{thm}

\begin{proof}  
It is easy to see that 
\[
({\bf U}^{2})^{+} = ({\bf U}^{+})^{2} + {\bf I}_{2m}
\]
for any simple graph $G$ 
with $\delta(G)\geq 3$; this equality is discussed also in 
\cite{GG2010, HKSS}. Then we have
\[
\det ( {\bf I}_{2m} -u ({\bf U}^{2})^+ )
= u^{2m}\cdot\det \Bigl(\frac{1-u}{u} {\bf I}_{2m} - ({\bf U}^{+} )^{2}\Bigr).
\]
Corollary~2.3 in our previous paper~\cite{HKSS} says that, 
for any $G$ with $\delta(G)\geq 2$, the following holds: 
\begin{align*}
\varphi(\lambda)&=\det \left( \lambda {\bf I}_{2m} - {\bf U}^{+} \right) \\
&= ( \lambda^{2} -1)^{m-n} \det \left( ( \lambda^{2} -1) {\bf I}_n 
- \lambda {\bf A}_{G} + {\bf D}_{G} \right). 
\end{align*}
Also refer to \cite{EmmsETAL2006,GG2010,KS2012}. 
It is easy to check  
\begin{align*}
\det &( {\bf I}_{2m} -u ({\bf U}^{2} )^{+} ) \\
&= u^{2m}\cdot\varphi(\sqrt{(1-u)/u})\cdot \varphi(-\sqrt{(1-u)/u})
\end{align*}
and 
\begin{align*}
&u^{n}\cdot\varphi(\sqrt{(1-u)/u}) = ((1-2u)/u)^{m-n}\cdot h_{G}(u),\\
&u^{n}\cdot\varphi(-\sqrt{(1-u)/u}) = ((1-2u)/u)^{m-n}\cdot l_{G}(u). 
\end{align*}
Combining the above, we can obtain the desired expression.
\end{proof}

Let us give information on a pole $u=1/2$, which is a final analogous part 
in Theorem~\ref{DEUZ} for the usual Ihara zeta function. 

Before stating the result, we introduce another kind of spanning graph 
in $G$ discussed in \cite{CDS2}: a spanning subgraph $H$ of $G$ is called 
an {\em odd-unicyclic factor} 
if each connected component of $H$ contains just one 
unoriented cycle of odd length and $V(H)=V(G)$. Here $H$ may not be connected, 
so we denote the number of components of $H$ by $\omega (H)$.
The terminology {\em unoriented cycle} here is the same as in Proof of 
Theorem~\ref{RC}. 
Moreover we write $OUCF(G)$ for the set of all odd-unicyclic factors in $G$. 

\begin{thm}\label{pole}
Let $G$ be a simple connected graph 
with $n$ vertices, $m$ unoriented edges and $\delta (G)\geq 3$.  
Set $p_{G}(u) = h_{G}(u)l_{G}(u)$ in Theorem~\ref{DE2}. 
Then $p_{G}(1/2)=0$. If $G$ is not bipartite, 
then the derivative at $u=1/2$ of $p_{G}(u)$ is as follows: 
\[
p'_{G} (1/2)= \frac{m-n}{2^{2n-2}} \cdot \kappa (G)\cdot\iota(G), 
\] 
where $\kappa (G)$ is the complexity of $G$ which is 
same as in Theorem~\ref{DEUZ} 
and $\iota (G)$ is the following graph invariant:
\[
\iota (G) =\sum_{H\in OUCF(G)}4^{\omega (H)}.
\]
On the other hand, if $G$ is bipartite, then $p'(1/2)=0$ 
and the second derivative at $u=1/2$ is as follows:
\[
p''_{G} (1/2)= \frac{(m-n )^{2} }{2^{2n-5}} ( \kappa (G))^{2}.
\]
\end{thm}

The following corollary is a direct consequence of Theorem~\ref{pole}.
\begin{cor}\label{order}
Let $G$ be a simple connected graph 
with $n$ vertices, $m$ unoriented edges and $\delta (G)\geq 3$.  
Then $u=1/2$ is a pole of the modified zeta function $\tilde{{\bf Z}}_{G}(u)$ 
whose order is $2(m-n+1)$ if $G$ is bipartite; $2(m-n)+1$ otherwise. 
\end{cor}

Before proving Theorem~\ref{pole}, we give some lemmas. 

\begin{lem}\label{lem1}
If $G$ is bipartite, then $h_{G}(u)=l_{G}(u)$.
\end{lem} 

\begin{proof}
It is  well known that 
$A$ and $-A$ are unitarily equivalent if $G$ is bipartite. 
In fact, let $V_{1}$ and $V_{2}$ be the bipartition of $V(G)$: 
$V(G)=V_{1}\sqcup V_{2}$. Then we put a diagonal matrix $T$ such that 
an $(i,i)$-element $T_{ii}=1$ if $v_{i}\in V_{1}$; otherwise $T_{ii}=-1$. 
It is easy to check that $A = T^{-1}(-A)T$. Therefore $h_{G}(u)=l_{G}(u)$ 
if $G$ is bipartite. 
\end{proof}

\begin{lem}\label{lem2}
Let $G$ be a simple connected graph with $n$ vertices.  
Then it holds that $h_{G}(1/2)=0$; $l_{G}(1/2)=0$ if $G$ is bipartite. 
Moreover, if $G$ is not bipartite, then $l_{G}(1/2)=2^{-n}\iota (G)$. 
\end{lem}

\begin{proof}
We can see that
$h_{G}(1/2)=2^{-n}\det ( {\bf D}_{G} - {\bf A}_{G})$ and 
$l_{G}(1/2)=2^{-n}\det ( {\bf D}_{G} + {\bf A}_{G})$. 
It is well known that ${\bf D}_{G} - {\bf A}_{G}$ is a discrete Laplacian 
and has $0$-eigenvalues. Thus $h_{G}(1/2)=0$. If $G$ is bipartite, 
it follows from Lemma~\ref{lem1} that $l_{G}(1/2)=0$. 
Theorem~4.4 in \cite{CDS2} tells us 
$\det ( {\bf D}_{G} + {\bf A}_{G}) = \iota (G)$. 
\end{proof}

\begin{proof}[Proof of Theorem~\ref{pole}]
For $p_{G}(u) = h_{G}(u)l_{G}(u)$, 
using Lemmas~\ref{lem1} and \ref{lem2}, we easily observe 
that, if $G$ is bipartite,
\[
p'(1/2)=2\cdot h_{G}(1/2)\cdot h'_{G}(1/2) = 0
\] 
and 
\begin{equation}\label{bip-2nd}
p''(1/2)=2\cdot (h'_{G}(1/2))^{2}.
\end{equation}
On the other hand, if $G$ is non-bipartite, 
\begin{eqnarray}\label{nonbip-1st}
p'(1/2)&=& h'_{G}(1/2)\cdot l_{G}(1/2) + h_{G}(1/2)\cdot l'_{G}(1/2) 
\nonumber\\
&=& 2^{-n}\iota (G)\cdot h'_{G}(1/2). 
\end{eqnarray}
Thus let us concentrate our attention on the computation on $h'_{G}(1/2)$. 
For $V(G)= \{ v_{1}, \cdots ,v_{n} \} $, 
we write $a_{i,j}$ for $(i,j)$-element of ${\bf A}_{G}$ 
and the matrix $M(u)$ for 
\[
{\bf I}_{n} - \sqrt{u(1-u)} {\bf A}_{G}+u ( {\bf D}_{G} -2 {\bf I}_{n} ). 
\] 
In addition, let us denote 
the derivative of the $(i,j)$-element of $M(u)$ by $m'_{i,j}(u)$ and 
the $(i,j)$-cofactor of $M(u)$ by $M_{i,j}(u)$. Here we remark that 
$M_{i,j}(1/2)$ coincides 
with the $(i,j)$-cofactor of $(1/2)({\bf D}_{G}-{\bf A}_{G})$; 
by the Matrix-Tree Theorem (\cite{Biggs, CDS}, for instance), we have 
\[
M_{i,j}(1/2) = \frac{1}{2^{n-1}} \kappa (G). 
\]
Furthermore, remarking that
\[
m'_{i,j}(u) =  - \frac{1-2u}{2 \sqrt{u(1-u)}} a_{i,j}+
( \deg v_{i} -2) \delta_{i,j},
\]
we easily obtain 
\begin{eqnarray}
h'_{G}(1/2)&=&\sum_{i,j}m'_{i,j}(1/2) M_{i,j}(1/2)\\\nonumber
&=&\frac{1}{2^{n-1}} \kappa (G)\sum_{i}( \deg v_{i} -2) 
= \frac{m-n}{2^{n-2}} \kappa (G). \nonumber
\end{eqnarray}
This completes the proof of Theorem~\ref{pole}. 
\end{proof}
\section{Example: distribution of poles of the modified zeta function}

Throughout this section, we assume a graph $G$ is $k$-regular with $n$ vertices
 and $m$ unoriented edges: $2m=kn$. Suppose further $k\geq 3$. 

For regular graphs, Theorem~\ref{DEUZ} was originally obtained 
by \cite{Ihara1966} in the context of a $p$-adic analogue of the Selberg 
zeta function. The concrete form in an analytic continuation 
from Theorem~\ref{DEUZ} is as follows: 
for a $k$-regular connected graph $G$ with $n$ vertices, 
\begin{equation}
{\bf Z}_{G}(u) = (1- u^{2} )^{n-kn/2}
\det ( {\bf I}_{n} -u {\bf A}_{G}+ (k-1)u^{2} {\bf I}_{n} )^{-1}. 
\end{equation}
Thus, in terms of eigenvalues of the adjacency matrix ${\bf A}_{G}$, 
we know the distribution of poles of ${\bf Z}_{G}(u)$. 
See \cite{Ihara1966,Sunada1986,Hashimoto1989,Bass1992}. 
Consequently, all of the real poles $u$ satisfy $1/(k-1)\leq |u|\leq 1$ 
and all of the imaginary poles $u$ lie on the circle 
whose center is the origin and radius is $1/\sqrt{k-1}$. 
Moreover it is concluded that  $u=1/(k-1)$ is a simple pole and 
$u=-1/(k-1)$ is also a simple pole if and only if $G$ is bipartite. 
As is stated in Theorem~\ref{DEUZ}, $u=1$ is a pole of order $(kn-2n+2)/2$. 
Usually the pole with $|u|=1$ or $1/(k-1)$ is called a {\em trivial pole}. 
If $G$ is a {\em Ramanujan graph}, that is, 
any nontrivial eigenvalue $ \lambda \neq \pm k$ of ${\bf A}_{G}$ satisfies 
$| \lambda | \leq 2 \sqrt{k-1} $, then any real pole is only {\em trivial pole} 
and any other poles lie on the circle above. 
In this sense, we say that 
the analogue of the Riemann hypothesis of the Ihara zeta function holds for 
a regular graph $G$ if and only if $G$ is a Ramanujan graph. 

We shall investigate the distribution of poles of 
the modified zeta function $\tilde{\bf Z}_{G}(u)$ for $k$-regular graphs.
Also in this case, 
in terms of eigenvalues of the adjacency matrix ${\bf A}_{G}$, 
we know the distribution of poles of $\tilde{\bf Z}_{G}(u)$. 
In particular, the eigenvalues of $({\bf U}^{2})^{+}$ are expressed 
by means of those of the adjacency matrix ${\bf A}_{G}$ of $G$   
in \cite{EmmsETAL2006,GG2010,HKSS} as follows:

\begin{thm}\label{51}(\cite{EmmsETAL2006})
Let $G$ be a simple connected $k$-regular graph 
with $n$ vertices and $m$ unoriented edges. 
Suppose that $k\geq 3$. 
The positive support $({\bf U}^2 )^{+} $ has $2n$ eigenvalues $\lambda_{2+}$
of the form 
\[
\lambda_{2+} = \frac{\lambda^{2}_{A} -2k+4}{2} 
\pm \sqrt{-1} \lambda_{A} \sqrt{k-1- \lambda^{2}_{A}/4}, 
\]
where $\lambda_{A} $ is an eigenvalue of the adjacent matrix ${\bf A}_{G}$. 
The remaining $2(m-n)$ eigenvalues of ${\bf U}^+$ are $2$. 
\end{thm}
By Proposition~\ref{PSDE}, 
an analytic continuation $\tilde{{\bf Z}}_{G}(u)$ has the following 
determinant expression: 
\begin{eqnarray}
\tilde{{\bf Z}}_{G}(u)&=& 1/\det ( {\bf I}_{2m} -u ({\bf U}^2 )^+ )\nonumber\\
&=& \prod_{ \lambda_{2+} \in Spec(({\bf U}^{2})^{+} )} (1-u \lambda_{2+} )^{-1};
\nonumber 
\end{eqnarray}
the poles of $\tilde{{\bf Z}}_{G}(u)$ is given by 
$1/ \lambda_{2+}$ for $\lambda_{2+} \in Spec(({\bf U}^{2})^{+} )$. 
Using Theorem~\ref{51}, we see the pole $u$ corresponding to $\lambda_{A}$ 
has the following form: 
\begin{equation}\label{form}
u=\frac{\lambda^{2}_{A}-2k+4 
\pm \sqrt{-1} \lambda_{A} \sqrt{4k-4-\lambda^{2}_{A} }}
{2( \lambda^{2}_A +(k-2)^{2} )}.
\end{equation}
Remarking that $u=1/(k^{2}-2k+2),1/2$, say {\em trivial poles}, 
if $\lambda_{A}=\pm k$ and 
$u=-1/(k-2)$ if $\lambda_{A}=0$, we can see the real poles 
$u\in [1/(k^{2}-2k+2),1/2]\cup \{-1/(k-2)\}$. 
Moreover it can be easily checked that 
any imaginary pole $u=p+q\sqrt{-1}$ $(p,q\in \mathbb{R})$ satisfies 
that 
\[
\left(p+\frac{1}{k^{2}-2k}\right)^{2}+q^{2} =
\left(\frac{k-1}{k^{2}-2k}\right)^{2}. 
\]

Let us summarize the above. 

\begin{exa}
Let $G$ be a simple connected $k$-regular graph with $n$ vertices. 
Suppose that $k\geq 3$. 
Then the pole of the modified zeta function $\tilde{{\bf Z}}_{G}(u)$ 
has the form as in (\ref{form}) with an eigenvalue 
$\lambda_{A}$ of the adjacency matrix ${\bf A}_{G}$. 
In particular, all of the real poles $u$ satisfy 
\[
\frac{1}{k^{2}-2k+2} \leq u\leq\frac{1}{2}
\]
and, if $0\in Spec({\bf A}_{G})$, $u=-1/(k-2)$; 
all of the imaginary poles $u$ lie on the circle whose center is 
$-1/(k^{2}-2k)$ and radius is $(k-1)/(k^{2}-2k)$. 
\end{exa}

Of course, we have already known in Theorem~\ref{RC} and Corollary~\ref{order}
$u=1/(k^{2}-2k+2)$ is a pole whose order is 2 or 1 
if $G$ is bipartite or not, respectively; 
 $u=1/2$ is a pole and 
its order is $(k-2)n+2$ or $(k-2)n+1$ if $G$ is bipartite or not, 
respectively. 
We should remark, for this modified zeta function $\tilde{{\bf Z}}_{G}(u)$, 
all poles except {\em trivial poles} lie on the circle above 
if $G$ is a {\em Ramanujan graph}. 
In this sense, we can say $\tilde{{\bf Z}}_{G}(u)$ also has 
a property of the analogue of the {\em Riemann hypothesis}.
%

\begin{figure}[h]
  \begin{center}
   \includegraphics[width=140mm]{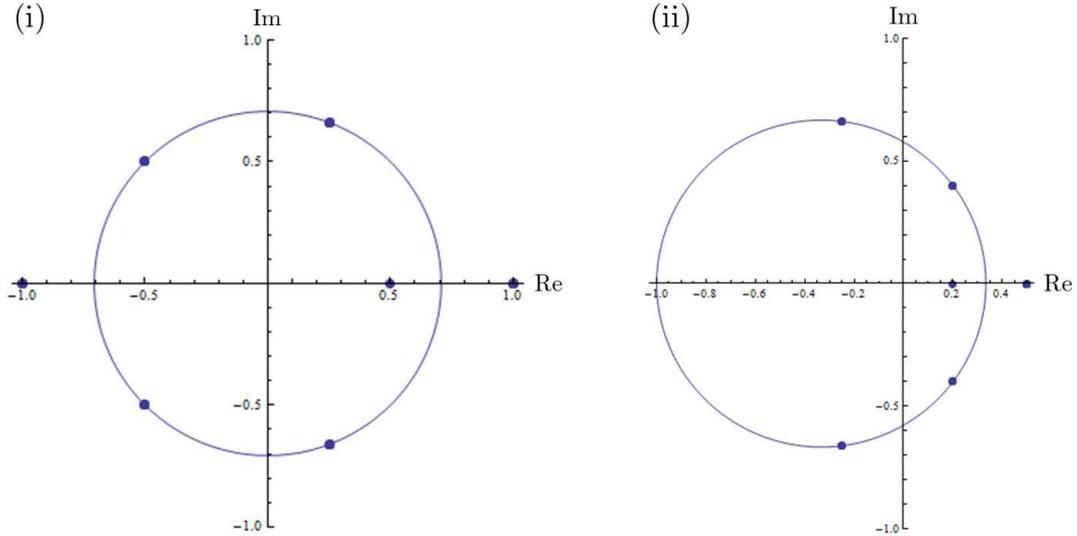}
  \end{center}
  \caption{
  {\small Poles of ${\bf Z}_G(u)$ and $\tilde{\bf Z}_G(u)$. 
  The dots in Figs. (i) and (ii) are the poles of 
  ${\bf Z}_G(u)$ and $\tilde{\bf Z}_G(u)$ of the Petersen graph, respectively. 
  The circles in Figs. (i) and (ii) are $p^2+q^2=1/(k-1)$ and 
  $\left(p+1/(k^2-2k)\right)^{2}+q^{2} =\left((k-1)/(k^{2}-2k)\right)^2$ for $k=3$, respectively. 
  Since the Petersen graph is a Ramanujan graph, all poles except trivial poles lie on the circles. }}
  \label{fig:one}
\end{figure}
\par
\
\par\noindent
\noindent
{\bf Acknowledgments.}
\par
\noindent YuH's work was supported in part 
by JSPS Grant-in-Aid for Scientific Research (C)~25400208 
and (B)~24340031. 
NK and IS also acknowledge financial supports of 
the Grant-in-Aid for Scientific Research (C) 
from Japan Society for the Promotion of Science 
(Grant No.~24540116 and No.~23540176, respectively).
ES thanks to the financial support of
the Grant-in-Aid for Young Scientists (B) of Japan Society for the
Promotion of Science (Grant No. 25800088).
\par

\begin{small}
\bibliographystyle{jplain}

\begin{thebibliography}{99}

\bibitem{Ambainis2003} 
Ambainis, A.: 
Quantum walks and their algorithmic applications, 
Int. J. Quantum Inf. {\bf 1}, 507--518 (2003).

\bibitem{AmbainisEtAl2001} 
Ambainis, A., Bach, E., Nayak, A., Vishwanath, A. and Watrous, J.: 
One-dimensional quantum walks. In: Proc. 33rd Annual ACM Symp. 
Theory of Computing, 37--49 (2001).

\bibitem{Bass1992}
Bass, H.: 
The Ihara-Selberg zeta function of a tree lattice, 
Internat. J. Math. {\bf 3}, 717--797 (1992). 

\bibitem{Biggs}
Biggs, N.:
Algebraic Graph Theory,
Cambridge Univ. Press, Cambridge, UK, 1974.

\bibitem{segawa8}
Chandrashekar, C. M., Banerjee, S. and Srikanthm R.: 
Relationship between quantum walk and relativistic quantum mechanics, 
Phys. Rev. A {\bf 81}, 062340 (2010).

\bibitem{segawa5}
Childs, A. M.: Universal computation by quantum walk, 
Physical Review Letter {\bf 102} 180501 (2009). 

\bibitem{CDS}
Cvetkovi\'{c}, D. M., Doob, M. and Sachs, H.:
Spectra of Graphs,
Academic Press, New York, 1979.

\bibitem{CDS2}
Cvetkovi\'{c}, D., Rowlinson, P. and Simi\'{c}, S. K.:
Signless Laplacians of finite graphs, 
Linear Algebra Appl. {\bf 423}, 155--171 (2007). 

\bibitem{EmmsETAL2006} 
Emms, D., Hancock, E. R., Severini, S. and Wilson, R. C.: 
A matrix representation of graphs and its spectrum as a graph invariant, 
Electr. J. Combin. {\bf 13}, R34 (2006). 
  
\bibitem{EmmsETAL2009} 
Emms, D.,  Severini, S., Wilson, R. C. and Hancock, E. R.: 
Coined quantum walks lift the cospectrality of graphs and trees, 
Pattern Recognition {\bf 42}, 1988--2002 (2009). 

\bibitem{GambelETAL}
Gamble, J. K., Friesen, M., Zhou, D., Joynt, R. and Coppersmith, S. N.:  
Two particle quantum walks applied to the graph isomorphism problem, 
Phys. Rev. A {\bf 81}, 52313 (2010).  

\bibitem{Ga}
Gantmacher, F. R.:
Theory of Matrices, 2 vols, Chelsea Publishing Co., 1959. 

\bibitem{GG2010} 
Godsil, C. and Guo, K.: Quantum walks on regular graphs and eigenvalues, 
Electron. J. Combin. {\bf 18}, P165 (2011).

\bibitem{GR} 
Godsil, C. and Royle, G.: 
Algebraic Graph Theory.
Springer-Verlag, New York, 2001.

\bibitem{Grover1}
Grover L. K.: 
A fast quantum mechanical algorithm for database search, 
Proceedings, 28th Annual ACM Symposium on the Theory of Computing, 
212 (1996).  

\bibitem{Grover2}
Grover L. K.: 
From Schr\"odinger's equation to quantum search algorithm, 
American Journal of Physics {\bf 69}, 769-777 (2001).

\bibitem{Hashimoto1989}
Hashimoto, K.:
Zeta Functions of Finite Graphs and Representations 
of $p$-Adic Groups,  
Adv. Stud. Pure Math. {\bf 1}5, 211--280 (1989).

\bibitem{HKSS}
Higuchi, Yu., Konno, N., Sato, I. and Segawa, E.:  
A note on the discrete-time evolutions of quantum walk on a graph,  
J. Math-for-Ind.  {\bf 5B}  (2013), 103--109. 

\bibitem{Ihara1966}
Ihara, Y.: 
On discrete subgroups of the two by two projective linear group 
over $p$-adic fields,  
J. Math. Soc. Japan {\bf 18}, 219--235 (1966). 

\bibitem{segawa18}
Karski, M., F\"oster, L., Choi, J.-M., Steffen, A., Alt, W., 
Meschede, D. and Widera, A.: 
Quantum walk in position space with single optically trapped atoms, 
Science {\bf 325}, 174 (2009).

\bibitem{Kempe2003} 
Kempe, J.: 
Quantum random walks - an introductory overview, 
Contemporary Physics {\bf 44}, 307--327 (2003). 

\bibitem{Konno2002} 
Konno, N.: Quantum random walks in one dimension, 
Quantum Inf. Process. {\bf 1}, 345--354 (2002).

\bibitem{Konno2005} 
Konno, N.: 
A new type of limit theorems for the one-dimensional quantum random walk, 
J. Math. Soc. Japan {\bf 57}, 1179--1195 (2005).

\bibitem{Konno2008b} 
Konno, N.: Quantum Walks. In: Lect. Notes Math.: Vol. 1954, 
309--452 (2008). 

\bibitem{KS2012}
Konno, N. and Sato, I.:  
On the relation between quantum walks and zeta functions. 
Quantum Inf. Process. {\bf 11} 341--349 (2012).

\bibitem{KS2000}
Kotani, M. and Sunada, T.:
Zeta functions of finite graphs. 
J. Math. Sci. U. Tokyo {\bf 7}, 7--25 (2000). 

\bibitem{segawa4}
Magniez, F., Nayak, A., Roland, J. and Santha, M.: 
Search via quantum walk, 
Proc. 39th ACM Symposium on Theory of Computing
575--584 (2007). 

\bibitem{manoucheheri2014}
Manouchehri, K. and  Wang, J.: 
Physical Implementation of Quantum Walks, 
Quantum Science and Technology, Springer, 2014. 

\bibitem{segawa17}
Matsuoka L. and Yokoyama, K.:  
Physical implementation of quantum cellular automaton in a diatomic molecule, 
J. Comput. Theor. Nanosci.: Special Issue:
``Theoretical and Mathematical Aspects of the Discrete Time Quantum Walk'' 
{\bf 10}, 1617--1620 (2013).

\bibitem{segawa6}
Mohseni, M., Rebentrost, P., Lloyd S. and Aspuru-Guzik, A.: 
Environment-assisted quantum walks in photosynthetic energy transfer, 
J. Chem. Phys. {\bf 129}, 174106 (2008).  

\bibitem{North}
Northshield, S.:
A note on the zeta function of a graph,
J. Combin. Theory Ser. B {\bf 74}, 408--410 (1998). 

\bibitem{RenETAL} 
Ren, P., Aleksic, T., Emms, D., Wilson, R. C. and Hancock, E. R.: 
Quantum walks, Ihara zeta functions and cospectrality in regular graphs, 
Quantum Inf. Proc. {\bf 10}, 405--417 (2011). 

\bibitem{segawa19}
Schreiber, A., Cassemiro, K. N., Poto\v{c}ek, V., G\'{a}bris, 
A., Mosley, P. J., Anderson, E., Jex, I. and Silberhorn, Ch.: 
Photons walking the line: A quantum walk with adjustable coin operations, 
Phys. Rev. Lett. {\bf 104}, 050502 (2010).

\bibitem{Segawa} 
Segawa,~E.: 
Localization of quantum walks induced by recurrence properties of random walks,
J. Comput. Theor. Nanosci.: Special Issue:
``Theoretical and Mathematical Aspects of the Discrete Time Quantum Walk''  
{\bf 10}, 1583--1590 (2013).   

\bibitem{ShiauETAL}
Shiau, S. -Y., Joynt, R. and Coppersmith, S. N..: 
Physically-motivated dynamical algorithms for the graph isomorphism problem. 
Quantum Inform. Comput. {\bf 5}, 492--506 (2005). 

\bibitem{Sm}
Smilansky, U.:
Quantum chaos on discrete graphs, 
J. Phys. A: Math. Theor. {\bf 40}, F621-F630 (2007). 

\bibitem{segawa7}
Strauch, F. W.: 
Connecting the discrete- and continuous-time quantum walks, 
Phys. Rev. A {\bf 74} 030301 (2006).

\bibitem{Sunada1986}
Sunada, T.: 
$L$-Functions in Geometry and Some Applications, 
In: Lect. Notes  Math. Vol. 1201, 266--284 (1986).

\bibitem{segawa3}
Szegedy, M.: Quantum speed-up of Markov chain based algorithms,
Proc. 45th IEEE Symposium on Foundations of Computer Science, 
32--41 (2004), 

\bibitem{VA} 
Venegas-Andraca, S. E.:  
Quantum walks: a comprehensive review, 
Quantum Inf. Process. {\bf 11}, 1015--1106 (2012). 

\bibitem{segawa20}
Z\"ahringer, F., Kirchmair, G., Gerritsma, R., Solano, E., Blatt, R. and Roos, C. F.: 
Realization of a quantum walk with one and two trapped ions, 
Phys. Rev. Lett. {\bf 104}, 100503 (2010). 

\end{thebibliography}

\end{small}

\end{document}